\documentclass[11pt]{article}
\usepackage{fullpage}

\usepackage{amsmath,amssymb,graphicx,amsthm}

\usepackage{color}

\newcommand{\ignore}[1]{}

\newcommand{\delete}[1]{}

\newtheorem{thm}{Theorem}
\newtheorem{lemma}[thm]{Lemma}

\newtheorem{definition}[thm]{Definition}
\newtheorem{cor}[thm]{Corollary}

\newcommand{\dist}{\ensuremath{{\sf dist}_{p, D}}}
\newcommand{\Astar}{\ensuremath{A^*}}
\newcommand{\G}{\ensuremath{\mathbb{G}}}

\newcommand{\Z}{\ensuremath{\mathbb{Z}}}
\newcommand{\Zm}{\ensuremath{\mathbb{Z}_N}}
\newcommand{\ZN}{\ensuremath{\mathbb{Z}_N}}
\newcommand{\Zp}{\ensuremath{\mathbb{Z}_p}}

\newcommand{\eps}{\ensuremath{\epsilon}}

\newcommand{\bool}{\ensuremath{\{0,1\}}}

\newcommand{\SumDist}{\textsc{Sum-Distinguish}}
\newcommand{\SumEqual}{\textsc{Sum-Equal}}

\newcommand{\deq}{\stackrel{{\rm def}}{=}}

\begin{document}

\title{\bf One-Round Multi-Party Communication Complexity of Distinguishing Sums}
\date{}
\author{
Daniel Apon\thanks{Dept.\ of Computer Science, University of Maryland. Email:
{\tt \{dapon,jkatz,amaloz\}@cs.umd.edu}.}
\and Jonathan Katz$^*$
\and Alex J.\ Malozemoff$^*$
}

\maketitle

\begin{abstract}
We consider an instance of the following problem:
Parties $P_1, \ldots, P_k$ each receive an input~$x_i$,
and a coordinator (distinct from each of
these parties)
wishes to compute $f(x_1, \ldots, x_k)$ for some predicate~$f$.
We are interested in \emph{one-round} protocols where each party sends a single
message to the coordinator; there is no communication between
the parties themselves.
What is the minimum communication complexity needed to compute~$f$, possibly with bounded error?

We prove tight bounds on the one-round communication complexity when $f$
corresponds to the promise problem of \emph{distinguishing sums} (namely,
determining which of two possible values the $\{x_i\}$ sum to) or the problem
of determining whether the $\{x_i\}$ sum to a particular value.  Similar
problems were studied previously by Nisan and in concurrent work by Viola.  Our
proofs rely on basic theorems from additive combinatorics, but are otherwise
elementary.
\end{abstract}

\section{Introduction}\label{sec:outline}

Consider the following general problem:
There are $k$ parties $P_1, \ldots, P_k$, with each party $P_i$ holding input~$x_i$.
A central coordinator (distinct from each of the parties)
wants to learn~$f(x_1, \ldots, x_k)$ for some fixed boolean function (or partial function)~$f$.
We are interested in \emph{one-round} protocols where each party sends a single
message to the coordinator and the coordinator then computes the result;
there is no communication between
the parties, nor does the coordinator send anything to the parties.
A trivial solution, of course, is for each party $P_i$ to send $x_i$ to the coordinator,
who then applies $f$ to the complete set of inputs and thus obtains the correct result.
For which functions $f$ can the total communication complexity be reduced,
possibly with bounded error?

Let $\G$ denote an abelian group and assume each party's input lies
in~$\G$. We study the communication complexity of two (related)
functions in the model described above.

\begin{definition}\label{def:distinguish}
Fix distinct $g_0, g_1 \in \G$.
The $k$-party \SumDist\ problem (relative to $g_0, g_1$) is defined by letting $f$ be the partial
function given by
\[f(x_1, \ldots, x_k) = \left\{ \begin{array}{ccl}
1 & & \mbox{if \; $\sum_i x_i = g_1$} \\
0 & & \mbox{if \; $\sum_i x_i = g_0$}
\end{array}\right..\]
\end{definition}

\begin{definition}\label{defn:sum-equal}
Fix $g \in \G$. The $k$-party
\SumEqual\ problem (relative to~$g$) is defined by letting $f$ be the function given by
\[f(x_1, \ldots, x_k) = \left\{ \begin{array}{ccl}
1 & & \mbox{if \; $\sum_i x_i = g$} \\
0 & & \mbox{otherwise}
\end{array}\right..\]
\end{definition}

We explore the communication complexity of solving the above
for $\G=\Z_p$ ($p$ prime) and
$\G=\Z$.\footnote{In the first case, each party's input is an
arbitrary element of~$\Z_p$; in the second case, each party's input
is an $n$-bit integer, with $n$ being an additional parameter of the
problem.}
Our proofs rely on the generalized
Cauchy-Davenport Theorem, but otherwise use only
elementary arguments.
Our results can be summarized as follows:
\begin{itemize}
\item For \SumDist\ with $\G=\Z$ or $\G=\Z_p$, $p$ prime, we show a deterministic protocol
with total communication complexity $k \log k + O(k)$; note that the communication
in the latter case is independent of~$p$.
For $\G=\Z_p$, we prove a lower bound of
$k \cdot \min\{\log k, \log p\} - k$ on the communication complexity of any deterministic protocol.

\item For \SumEqual\ with $\G=\Z$ or $\G=\Z_p$, $p$ prime,
we show a protocol using public randomness
with error~$\epsilon$ and total communication complexity $k \log k/\eps + O(k)$.
A lower bound (for deterministic protocols and $\G=\Zp$) is implied by
our lower bound for \SumDist.
\end{itemize}
We also briefly consider the case $\G=\Z_N$ for square-free $N$.

\subsection{Motivation}
The problems above are natural in the \emph{number-in-hand} model of
multi-party communication complexity, and variants of the \SumDist\
and \SumEqual\ problems have been considered in prior
work~\cite{Smi88, Nisan93, B:KusNis97, MNSW98, ECCC:Viola11}, sometimes
for $k=2$ only. (We survey prior results in the next section.)

Our motivation, though, comes from the domain of \emph{distributed
intrusion detection}. The goal of distributed intrusion-detection
systems (DIDS) is to monitor a network across a number of
hosts in order to detect aberrant
behavior (indicating a potential intrusion) and, if detected, raise an alarm. In typical
operation of DIDS, each host
records some observations over a specified time period; at the end of this period,
each of those hosts sends all the data it
has recorded to a central coordinator, which then determines
--- based on the aggregate data from all the hosts --- whether or not to issue an
alarm. In some systems (e.g., when the hosts are geographically distributed, when
communication is over a low-bandwidth channel, and/or when the volume of data recorded
at each host is huge), reducing the communication becomes critical.
While there has been some work aimed at reducing the communication complexity of
DIDS~\cite{SRImalware,SRIalert,abstraction}, we are not aware of any prior theoretical study
of the problem.

If we model the decision of the coordinator by some predicate~$f$ computed over the data $x_1, \ldots, x_k$
recorded by each host, we recover exactly the general problem being considered here.
(For the application to distributed intrusion detection,
direct communication between the hosts would typically
be impossible, and it would be undesirable for the coordinator to have to send data to the hosts.)
While \SumDist\ and \SumEqual\ are too simplistic to capture real-world decision procedures,
they were chosen to
correspond to the ``DIDS-like'' problems of distinguishing between a ``good'' system state
$g_0$ and a ``bad'' system state~$g_1$ (in the
case of \SumDist),
or identifying when the system is in one particular ``bad'' state~$g$
(in the case of \SumEqual).

\subsection{Prior Work}

For the case of \SumEqual\ with $\G=\Z$ and where each party's
input is an $n$-bit integer, Nisan~\cite{Nisan93} shows a randomized protocol with
total communication complexity $O(k \log n)$.
Our deterministic protocol achieves better communication complexity $k\log k + O(k)$
when $k < n$.
%
%
In concurrent and independent work, Viola~\cite{ECCC:Viola11}
studies $\SumEqual$ with $\G=\Z_p$,
and shows $\Theta(k \log k)$ upper and lower bounds
on the communication complexity for certain ranges of~$k$ and~$p$.
Our protocols for \SumEqual\ achieve similar bounds
for more general $k, p$, and using different tools.

\ignore{From Daniel: Belated comparison to Viola's paper:
\begin{quote}\begin{em}
Viola's SUM-EQUAL in $\Zp$, $p$ prime, lower bound is:

``The CC of $k$-player SUM-EQUAL modulo $p$ a prime between $k^{(1/4)}$ and
$2k^{(1/4)}$ is $\Omega(k \log k)$.''

So, he's using $p < 2k^{(1/4)} < k$. We're doing $k < p$, a different case.

His upper bound is for the case $k < p$.\end{em}
\end{quote}}

Our protocols use a direct, combinatorial perspective that (along the way)
explores a new connection between communication complexity and additive combinatorics
that may be appealing in its own right. It will be interesting
to explore other connections between these fields.



\subsection{Organization}

In Section~\ref{sec:prelim}, we recall the necessary preliminaries from
additive combinatorics. In Section~\ref{sec:sumdist}, we prove upper and lower
bounds for the \SumDist\ problem over~\Zp. In
Section~\ref{sec:sumequal}, we give a randomized protocol for
\SumEqual\ over \Zp. In Section~\ref{sec:extensions}, we show how our protocols can
be extended
to work over \Z\ or \Zm\ for $m$ the product of few primes.

\section{Preliminaries}\label{sec:prelim}

We let \G\ denote an abelian group, written additively.
$\Z$ denotes the integers, and $\Z_p$ is the group $\{0, \ldots, p-1\}$
under addition modulo~$p$.
We use ``$\log$'' to refer to logarithms base~2.

\subsection{Tools from Additive Combinatorics}

We utilize two well-studied, fundamental objects
from additive combinatorics:
{\em sumsets\/} and {\em arithmetic progressions}.

\begin{definition}\label{defn:sumset}
For (not necessarily distinct) sets $A_1,\dots,A_k \subseteq \G$, define their {\sf sumset} as
$\sum_{i=1}^k A_i=A_1+\dots+A_k\deq\left\{\sum_{i=1}^k a_i\;|\;a_1 \in A_1,\dots, a_k \in A_k\right\}$. That
is, $\sum_i A_i$ is the set of all possible sums obtainable
by choosing one element from each set~$A_i$.
\end{definition}

In our constructions we use sumsets of arithmetic progressions, i.e.
sequences of integers with common difference~$D$.
We refer to these as~\emph{$D$-APs}.

\begin{definition}\label{defn:ap-set}
Fix a prime $p$ and a difference $D\ne 0\bmod p$. For any $b\in \{0, \ldots, D-1\}$,
let
\[
A_{(b)}\deq\left\{b,\; b+D,\; b+2D,\; \dots,\; b+\left(\left\lfloor\frac{p-1-b}{D}\right\rfloor\right)D\right\}\subseteq\Zp
\]
denote the {\sf $D$-AP in $\mathbb{Z}_p$ with base~$b$}.
\end{definition}

\noindent
Note that we only
consider $D$-APs of
maximal size with no ``wrap-around'';
i.e., the base $b$ is less than $D$,
and the progression contains $b+iD$ for all $i \geq 0$ with
$b+iD<p$.
As an example, the maximal 7-APs in $\Z_{19}$ are
\begin{eqnarray*}
& A_{(0)} = \{0, 7, 14\}; \;\;\; A_{(1)} = \{1,8,15\}\; \;\;\; A_{(2)} = \{2,9,16\} & \\
& A_{(3)} = \{3,10,17\};\;\;\; A_{(4)} = \{4,11,18\};\;\;\;A_{(5)} = \{5,12\};\;\;\; A_{(6)} = \{6,13\}.&
\end{eqnarray*}



For our lower bounds, we use the generalized Cauchy-Davenport Theorem~\cite{JEP:Cauchy1813,JLM:Davenport1935}.

\begin{thm}[Generalized Cauchy-Davenport Theorem]\label{thm:c-d-general}
For a prime $p$, and $k$ (not necessarily distinct) nonempty
sets $A_1, \dots, A_k\subseteq \Zp$,
\[
\left|{\textstyle \sum_{i=1}^k A_i} \right|\geq\min\left\{p, \;
{\textstyle \sum_{i=1}^k\left|A_i\right|}-k+1\right\}.
\]
\end{thm}

\noindent
For our constructions, we rely on the fact that sumsets of $D$-APs achieve the above minimum.

\begin{lemma}\label{lemma:c-d-gen-min}
For a prime $p$,
and $k$ (not necessarily distinct) $D$-APs $A_1,\ldots,A_k \subseteq \Zp$,
\[
\left|{\textstyle \sum_{i=1}^k A_i}\right|=\min\left\{p, \; {\textstyle \sum_{i=1}^k |A_i|}-k+1\right\}.
\]
\end{lemma}

\begin{proof} We prove the lemma for $k=2$; the general case
follows by induction. Let $A, B$ denote the two sets in question. By
Theorem~\ref{thm:c-d-general}, we have
$|A+B| \ge \min\{p, |A|+|B|-1\}$. It remains to upper bound $|A+B|$.

Write $A = \{b_A + i  D \mid 0 \leq i < |A|\}$,
$B = \{b_B + i' D \mid 0 \leq i' < |B|\}$
with $0\le b_A, b_B < D$. Then
\begin{eqnarray*}
A+B &=& \left\{\rule{0pt}{9pt} b_A + b_B + Di + Di'\bmod{p} \mid 0 \leq i < |A|, \, 0 \leq i' < |B| \right\} \\
    &=& \left\{\rule{0pt}{9pt} b_A + b_B + Di''\bmod{p} \mid 0 \leq i'' \leq |A|+|B|-2\right\}.
\end{eqnarray*}
So $|A+B|$ contains at most $|A|+|B|-1$ elements, giving the desired bound.
\end{proof}

\subsection{Notions of Distance and Contiguity}\label{sec:distance}

In  Section~\ref{sec:ub-dist}, we use a notion of distance between two elements $g_0,g_1\in\Zp$.
Specifically,
we define their distance relative to some difference $D$
to be the minimum number of additions or subtractions by~$D$
(modulo~$p$) needed to map $g_0$ to~$g_1$.
We define this formally next.

\begin{definition}\label{defn:distance}
Fix a prime $p$ and a difference $D\ne 0\bmod p$. For any $g_0, g_1\in\Zp$,
define the {\sf distance from $g_0$
to~$g_1$} (relative to~$p, D$) as
\[
\dist(g_0, g_1) \deq \min\left\{(g_1 - g_0)D^{-1}\bmod{p},\; (g_0-g_1)D^{-1}\bmod{p}\right\}.
\]
(The minimum is taken by viewing each term as an integer in $\{0, \ldots, p-1\}$.)
\end{definition}


We say $g_0, g_1\in\Zp$ are \emph{adjacent} (with respect to $\dist$) if $\dist(g_0, g_1)=1$.
We say a set $A\subseteq\Zp$
is \emph{contiguous} (relative to~$D$) if it can be ordered so that all adjacent elements
in the ordering are adjacent with respect to~$\dist$. When $|A|=1$,
$A$ is vacuously contiguous. Clearly $D$-APs are contiguous; we observe that sumsets of $D$-APs are also contiguous.

\begin{lemma}\label{lemma:contiguous}
Fix a prime $p$, difference $D\ne 0\bmod p$, and any $D$-APs $A_1, \ldots, A_k \subseteq\Zp$.
Then $\sum_{i=1}^k A_i$ is contiguous relative to~$D$.
\end{lemma}
\begin{proof}
From the proof of Lemma~\ref{lemma:c-d-gen-min}, for any
$D$-APs $A$ and $B$ we have \[A+B = \{b_A + b_B + iD\bmod{p} \mid 0 \leq i \leq |A|+|B|-2\},\]
which is contiguous by definition.
Induction on $k$ completes the proof.
\end{proof}

\begin{cor}\label{cor:separate}
Fix a prime $p$, difference $D\ne 0\bmod p$, and $D$-APs $A_1, \ldots, A_k\subseteq\Zp$.
For any $g_0, g_1\in\Zp$,
if $\dist(g_0, g_1)\ge |\sum_{i=1}^k A_i|$ then $g_0$ and $g_1$ cannot both be in $\sum_{i=1}^k A_i$.
\end{cor}

Consider again the example of $\Z_{19}$ with $D=7$.
Then $A_{(2)}+A_{(3)}=\{2, 9, 16\}+\{3, 10, 17\}=\{5,12,0,7,14\}$ is contiguous, and
$|A_{(2)}+A_{(3)}|=5$.
Taking $2, 5 \in \Z_{19}$, we have
${\sf dist}_{19,7}(5,2)=5 \geq |A_{(2)}+A_{(3)}|$ and, indeed,
$5\in A_{(2)}+A_{(3)}$ but $2\not\in A_{(2)}+A_{(3)}$.

\section{\SumDist\ over \Zp}\label{sec:sumdist}
\subsection{A Deterministic Protocol}\label{sec:ub-dist}

Corollary~\ref{cor:separate} suggests a technique for
efficiently distinguishing two sums. Say $k$ parties wish to
determine whether their inputs $x_1, \ldots, x_k$ sum to $g_0$
or~$g_1$ (modulo~$p$). For some fixed, agreed-upon difference~$D$
(we discuss how to set~$D$ below), each party $P_i$ sends to the
coordinator the \emph{index} of the $D$-AP $A_i$ in which its input~$x_i$
lies. The coordinator thus learns that the sum $\sum_i x_i$ lies in
the sumset $A \deq \sum_{i=1}^k A_i$. As long $\dist(g_0, g_1) \geq
|A|$, it cannot be the case that both $g_0$ and $g_1$ are in~$A$; in
that case, the coordinator learns the sum by checking which of $g_0, g_1$ lies in~$A$.

The main difficulty in implementing the above is that $g_0, g_1$ may
be very ``close.'' In that case, in order to ensure that the above
succeeds we need to ensure that $|A|$ is small. This, in turn,
requires the $D$-APs to be small, which means that there are more of
them. Since the communication from each party is the logarithm of
the number of $D$-APs, this makes the communication complexity
worse. Ideally, we would like to set $D$ independently of the
relative distance between $g_0$ and~$g_1$.

A solution is to have the parties ``shift'' their inputs by each
locally multiplying them (modulo~$p$) by an agreed-upon
constant~$c$. The problem then reduces to distinguishing whether the
shifted inputs sum to $g'_0 \deq c\cdot g_0 \bmod p$ or $g'_1 \deq c
\cdot g_1 \bmod p$. The insight is that regardless of $g_0, g_1$, we
can set $c$ appropriately to ensure that $g'_0$ and $g'_1$ are ``far
apart.''

\delete{We suggest that the reader begin by examining the statements
of Lemmas~\ref{lemma:compute-c} and~\ref{lemma:compute-d} and
Protocol~\ref{prot:full} in order to understand the \emph{geometric
consequences} of each. We recommend referring to
Figure~\ref{fig:Z19-cycle} for a visual interpretation and the
accompanying discussion in Section~\ref{sec:distance}. A more
careful analysis of the detailed arithmetic in each proof is found
in the appendix and is better tackled after some initial intuition
for its purpose is gained.}

We proceed with the details, beginning with some preliminary lemmas.

\begin{lemma}\label{lemma:compute-c}
Fix a prime $p>2$ and a difference $D\ne 0\bmod p$. Then for any
distinct $g_0, g_1\in\Zp$ there exists a value $c \neq 0 \bmod p$
such that $\dist(c\cdot g_0\bmod p,\; c\cdot g_1\bmod
p)=\frac{(p-1)}{2}$.
\end{lemma}
\begin{proof}
Set $c=\frac{(p-1)}{2}D(g_1-g_0)^{-1}\bmod p$. Then
\begin{eqnarray*}
c(g_1-g_0)D^{-1} & = &
    \frac{(p-1)}{2}D(g_1-g_0)^{-1}(g_1-g_0)D^{-1}\bmod p \\
    &=&\frac{(p-1)}{2}\bmod p.
\end{eqnarray*}
Since $(p-1)/2 < -(p-1)/2 \bmod p$ (viewing the right-hand term as
an integer in $\{1, \ldots, p-1\}$), this completes the proof.
\end{proof}


\begin{lemma}\label{lemma:compute-d}
Fix a prime $p>5$, and integer $k < p/4$. Set
$D=\left\lceil\frac{2kp}{(p-3)}\right\rceil<p$. Then for any $D$-APs
$A_1, \dots, A_k\subseteq\Zp$, we have $\left|\sum_{i=1}^k
A_i\right|\le\frac{(p-1)}{2}$.
\end{lemma}
\begin{proof}
By Lemma~\ref{lemma:c-d-gen-min},
\[
\left|\sum_{i=1}^k A_i\right|=\min\left\{p,\; \sum_{i=1}^k\left|A_i\right|-k+1\right\}.
\]
Observe that $|A_i|\le\lceil\frac{p}{D}\rceil\le\frac{p}{D}+1$. Therefore,
\begin{eqnarray*}
\sum_{i=1}^k\left|A_i\right|-k+1 &\le& \sum_{i=1}^k\left(\frac{p}{D}+1\right)-k+1 \\
    &\le& \frac{kp(p-3)}{2kp}+k-k+1 = \frac{(p-1)}{2}, 
\end{eqnarray*}
completing the proof.
\end{proof}

\begin{thm}\label{theorem:ub-dist}
There is a universal constant $C$ such that for any prime $p$ and
positive integer~$k$ there is a $k$-party, one-round, deterministic protocol
for \SumDist\ over~\Zp\ having communication complexity $k \log k +
C \cdot k$.
\end{thm}
\begin{proof}
There is a trivial protocol having communication complexity $k \cdot
\lceil \log p \rceil$, so the theorem is trivially true if $p\leq 5$
or $k \geq p/4$. In what follows we therefore assume $p>5$ and
$k<p/4$.

Fix arbitrary, distinct $g_0, g_1 \in \Zp$. The protocol for solving
\SumDist\ relative to~$g_0, g_1$ is as follows. Set $c$ as in
Lemma~\ref{lemma:compute-c}, and $D$ as in
Lemma~\ref{lemma:compute-d}. Party $P_i$, holding input~$x_i$,
computes
\[
b_i = \left(\left(c\cdot x_i\right)\bmod p\right)\bmod D
\]
and sends $b_i$ to the coordinator. (Note that $b_i$ is the base for
the $D$-AP $A_{(b_i)}$ containing~$c \cdot x_i \bmod p$.)
The coordinator outputs 0 if $c\cdot g_0 \bmod p$ is in
$\sum_{i=1}^k A_{(b_i)}$, and outputs 1 otherwise.

If $\sum_i x_i = g_0\bmod p$ then $\sum_i c \cdot x_i = c \cdot g_0
\bmod p$, and it is immediate that the coordinator outputs the
correct answer~0. So, assume instead that $\sum_i x_i=g_1\bmod p$.
Then $c\cdot g_1\bmod p$ is in $\sum_i A_{(b_i)}$. Since
$\dist(c\cdot g_0\bmod p,c\cdot g_1\bmod p) = (p-1)/2$ (by
Lemma~\ref{lemma:compute-c}) and $\left|\sum_i A_{(b_i)}\right| \leq (p-1)/2$
(by Lemma~\ref{lemma:compute-d}), we conclude from
Corollary~\ref{cor:separate} that $c \cdot g_0 \bmod p$ is not in
$\sum_i A_{(b_i)}$. Hence, in this case the coordinator outputs the
correct answer~1.

The communication complexity is exactly $k \cdot \lceil \log D
\rceil$ bits. Since $D \leq \frac{2kp}{p-3}+1\leq C' \cdot k$ for
some constant~$C'$ independent of $p$ and~$k$, this completes the
proof.
\end{proof}

\medskip\noindent{\bf Efficient implementation.}
We note that the coordinator can be implemented to run
efficiently. (It is clear that the parties can run efficiently.) First note that from
$b_i$ the coordinator can efficiently compute $\left|A_{(b_i)}\right| = \lceil
\frac{p-1-b_i}{D} \rceil + 1$. It can then compute $d \deq \left|\sum_i A_{(b_i)}\right|
= \sum_{i=1}^k |A_{(b_i)}|-k+1$.
Finally,
the coordinator can check whether $c \cdot g_0 \in \sum_i A_{(b_i)}$ by
computing $b^* = \sum_i b_i\bmod p$ and then checking whether
$(c\cdot g_0-b^*)\cdot D^{-1}\bmod p$ is less than~$d$.

\subsection{A Lower Bound for Deterministic Protocols}

In the following, we consider one-round protocols in which each party always
sends exactly $t$ bits to the coordinator, for some~$t$. We say any such
protocol has \emph{per-party communication complexity~$t$}.

The basic idea of the lower bound is as follows. Each message~$m \in
\bool^t$ from party $P_1$, say, defines a set $A_{1,m}$ of possible
inputs~$x_1$ (namely, those inputs on which $P_1$ would send~$m$).
Given the messages $m_1, \ldots, m_k$ sent by all the parties, the
coordinator learns only that the sum $\sum_i x_i$ lies in the sumset
$\sum_i A_{i,m_i}$. If we can show that there exist some $m_1,
\ldots, m_k$ for which $\sum_i A_{i,m_i}$ contains both $g_0$
and~$g_1$, then there must be some set of inputs on which the
protocol outputs the wrong result. The crux of the proof is to show
that if $t$ is too small, then there exist $m_1, \ldots, m_k$ for
which $\sum_i A_{i,m_i}=\Zp$, and hence the sumset does indeed
contain both $g_0$ and~$g_1$.

\begin{thm}\label{theorem:dlb-dist}
Fix prime $p$ and positive integer $k>1$. If $t \leq \min\{\log
((k-1)/2), \log(p/2)\}$, there is no deterministic, $k$-party
protocol for \SumDist\ over \Zp\ (relative to {\sf any} $g_0, g_1
\in \Zp$) with per-party communication complexity~$t$.
\end{thm}
\begin{proof}
Fix some deterministic protocol for \SumDist\ over \Zp\ (relative to
some $g_0, g_1 \in \Zp$) with per-party communication
complexity~$t$. The protocol defines for each party~$P_i$ a
partition $A_{i,1}, \ldots, A_{i,2^t}$ of~\Zp, where $A_{i,j}$ is
the set of inputs which cause $P_i$ to send $j$ to the coordinator.

For each party $P_i$ there exists an $m_i$ such that $|A_{i,m_i}|
\geq p/2^t$. Moreover, there is a legal set of inputs for the
parties such that $P_1, \ldots, P_{k-1}$ send $m_1, \ldots,
m_{k-1}$, respectively. (Simply take $x_i \in A_{i,m_i}$ for $i=1,
\ldots, k-1$, and then let $x_k \in \{g_0-\sum_{i=1}^{k-1} x_i, \;
g_1-\sum_{i=1}^{k-1} x_i\}$.) When the coordinator receives $m_1,
\ldots, m_{k-1}$ then, even if it is additionally given $P_k$'s
input~$x_k$, the coordinator learns only that the sum $\sum_i x_i$
of the parties' inputs lies in the set $x_k + \sum_{i=1}^{k-1}
A_{i,m_i}$. By Theorem~\ref{thm:c-d-general}, however, we have
\[
\left|\sum_{i=1}^{k-1} A_{i,m_i} \right| \ge \min\left\{p, \;
\sum_{i=1}^{k-1}\frac{p}{2^t}-(k-1)+1\right\} = \min\left\{p,
\frac{(k-1)p}{2^t}-k+2\right\}.
\]

If $p \geq k-1$ then $t \leq \log ((k-1)/2)$ and
\[\frac{(k-1)p}{2^t}-k+2 \geq 2p-k+2 > p.\]
On the other hand, if $k-1 > p$ then $t \leq \log (p/2)$ and
\[\frac{(k-1)p}{2^t}-k+2 \geq k > p.\]
In either case, then, we must have $\left|\sum_{i=1}^{k-1} A_{i,m_i}
\right| \geq p$ and so $\sum_{i=1}^{k-1} A_{i,m_i} = \Zp$. This
implies that there exist inputs $x_1, x'_1 \in A_{1,m_1}, \ldots,
x_{k-1}, x'_{k-1} \in A_{k,m_{k-1}}$ with $x_k+\sum_{i=1}^{k-1} x_i
= g_0$ and $x_k+\sum_{i=1}^{k-1} x'_i = g_1$. But then there exists
some set of legal inputs for the parties on which the coordinator
outputs an incorrect result.
\end{proof}

\section{A Randomized Protocol for Sum-Equal over \Zp}\label{sec:sumequal}

Our protocol for \SumEqual\ is similar to
our protocol for \SumDist. Namely, each party $P_i$ scales its input $x_i$ by some value~$c$
and sends the index of the $D$-AP $A_i$ that contains the scaled value~$c \cdot x_i \bmod p$;
the coordinator outputs~1 iff $c \cdot g \in \sum_i A_i$.

Note that the coordinator never errs if $\sum_i x_i = g$, and so we need only
analyze the case when $\sum_i x_i \neq g$.
In the case of \SumDist, we are guaranteed that $\sum_i x_i \in \{g_0, g_1\}$
and so we set~$c$ to some fixed
value such that $cg_0$ and $cg_1$ are ``far apart.''
The problem here is that $\sum_i x_i$ can be arbitrary.
To deal with this, we have
the parties select~$c \in \Zp$
uniformly at random using the public randomness.
If $\sum_i x_i = g' \neq g$ then the protocol will succeed as long as $cg$ and $cg'$ are sufficiently
``far apart'' as before. By setting the parameters of the protocol appropriately,
we ensure that this happens with high probability over choice of~$c$.

\begin{lemma}\label{lemma:compute-c2}
Fix a prime $p>2$, a difference $D\ne 0\bmod p$, and $\xi \in (0,1)$.
Then for any
distinct $g, g' \in\Zp$ there are at least $\xi \cdot (p-1)$
values $c \neq 0 \bmod p$
such that \[\dist(c\cdot g\bmod p,\; c\cdot g' \bmod
p)> (1-\xi) \cdot \frac{ (p-1)}{2}.\]
\end{lemma}
\begin{proof}
Set $\delta \deq \lceil \xi \cdot (p-1)/2 \rceil$, and
take any $d$ in the set
$\left\{\frac{(p-1)}{2} - \delta+1, \ldots, \frac{(p-1)}{2}+\delta\right\}$
of size $2\delta$. Set
$c=d \cdot D(g-g')^{-1}\bmod p$. Then
\begin{eqnarray*}
c(g-g')D^{-1} =
    d \cdot D(g-g')^{-1}(g-g')D^{-1}
    = d \bmod p.
\end{eqnarray*}
So,
\begin{eqnarray*}
\dist(c\cdot g\bmod p,\; c\cdot g' \bmod p) & = & \min \{c(g-g')D^{-1} \bmod p, \;
c(g'-g)D^{-1} \bmod p\} \\
& = &
 \min\{d, p-d\} \;\; \geq \;\; \frac{(p-1)}{2} - \delta+1,\end{eqnarray*}
completing the proof.
\end{proof}

\begin{lemma}\label{lemma:compute-d2}
Fix a prime $p>5$, an integer $k< p/4$, and $\epsilon > \frac{2k}{p-3}$. Set
$D=\left\lceil\frac{2kp}{\epsilon (p-3)}\right\rceil<p$. Then for any $D$-APs
$A_1, \dots, A_k\subseteq\Zp$, we have $\left|\sum_{i=1}^k
A_i\right| < \epsilon \cdot \frac{(p-1)}{2}+1$.
\end{lemma}
\begin{proof}
By Lemma~\ref{lemma:c-d-gen-min},
\[
\left|\sum_{i=1}^k A_i\right|=\min\left\{p,\; \sum_{i=1}^k\left|A_i\right|-k+1\right\}.
\]
Observe that $|A_i|\le\lceil\frac{p}{D}\rceil\le\frac{p}{D}+1$. Therefore,
\begin{eqnarray*}
\sum_{i=1}^k\left|A_i\right|-k+1 &\le& \sum_{i=1}^k\left(\frac{p}{D}+1\right)-k+1 \\
    &\le& \frac{kp \epsilon (p-3)}{2kp}+k-k+1 < \frac{\epsilon(p-1)}{2} + 1, 
\end{eqnarray*}
completing the proof.
\end{proof}

\begin{thm}
There is a universal constant $C$ such that for any prime $p$,
positive integer~$k$, and $\epsilon \in (0, 1)$, there is a $k$-party, one-round protocol
for \SumEqual\ over~\Zp\ using public randomness, with error at most~$\epsilon$ and
communication complexity $k \log k/\epsilon + C \cdot k$.
\end{thm}
\begin{proof}
There is a trivial protocol with communication complexity $k \cdot
\lceil \log p \rceil$, so the theorem is true if $p\leq 5$
or $k \geq p/4$ or $\epsilon \leq \frac{2k}{p-3}$. In what follows we therefore assume $p>5$,
$k<p/4$, and $\epsilon > \frac{2k}{p-3}$.

The protocol for solving
\SumEqual\ is as follows. Set $D$ as in
Lemma~\ref{lemma:compute-d2}, and use the public randomness to choose uniform $c \in \Zp\setminus \{0\}$.
Party $P_i$, holding input~$x_i$,
computes
\[
b_i = \left(\left(c\cdot x_i\right)\bmod p\right)\bmod D
\]
and sends $b_i$ to the coordinator.
The coordinator outputs 1 if $c\cdot g \bmod p$ is in
$\sum_{i=1}^k A_{(b_i)}$, and outputs 1 otherwise.

If $\sum_i x_i = g\bmod p$ then $\sum_i c \cdot x_i = c \cdot g
\bmod p$ and the coordinator always outputs the
correct answer~1. Now say $\sum_i x_i=g' \neq g \bmod p$.
Then $c\cdot g'\bmod p$ is in $\sum_i A_{(b_i)}$.
Using Lemma~\ref{lemma:compute-d2}, we have
$\left|\sum_i A_{(b_i)}\right| < \epsilon (p-1)/2 + 1$.
Using Lemma~\ref{lemma:compute-c2},
with probability at least $\xi \deq 1-\epsilon$ we have
$\dist(c\cdot g\bmod p,\; c\cdot g'\bmod p) > \epsilon (p-1)/2$.
Assuming that to be the case, we have
\[\dist(c\cdot g\bmod p,\; c\cdot g'\bmod p) \geq \left|{\textstyle \sum_i A_{(b_i)}}\right|\]
(note that both sides of the above are integers), and so we
conclude from
Corollary~\ref{cor:separate} that $c \cdot g \bmod p$ is not in
$\sum_i A_{(b_i)}$. We thus see that with probability at least $1-\epsilon$ the coordinator outputs the
correct answer~0.

The communication complexity is exactly $k \cdot \lceil \log D
\rceil$ bits. Since $D \leq \frac{2kp}{\epsilon(p-3)}+1\leq C' \cdot k/\epsilon$ for
some constant~$C'$ independent of $p, k$, and $\epsilon$, this completes the
proof.
\end{proof}

\section{Protocols Over \Z\ and \Zm}\label{sec:extensions}

In what follows, we show how to modify our protocols to work over
the integers and in \ZN\ for square-free~$N$.

\medskip\noindent{\bf Protocol over \Z.} Working over \Z\ is relatively easy.
The parties are given inputs in $\{0, \ldots, 2^n-1\}$. The maximum
sum of all the inputs is~$k 2^{n}$, and the ``target values'' are at
most that also. The parties choose the smallest prime $p > k2^{n}$,
treat their inputs as lying in~$\Zp$, and run the protocol for \Zp.
Note that $\sum_i x_i = g$ over the integers iff $\sum_i x_i = g
\bmod p$ by our choice of~$p$.

\medskip\noindent{\bf Protocol over \Zm.}
Assume $N$ is square-free, and let $N = \prod_{i=1}^m p_i$ be the
prime factorization of~$N$. The parties can then work modulo each of
the~$p_i$, and rely on the Chinese remainder theorem for
correctness. The complexity of the protocol scales with the number
of prime factors~$m$.

\def\shortbib{0}


\bigskip\noindent
{\small This research was sponsored by the Army Research Laboratory
and was accomplished under
Cooperative Agreement Number W911NF-11-2-0086. The views and conclusions contained in this document
are those of the authors and should not be interpreted as representing the official policies,
either expressed or implied, of the Army Research Laboratory or the U.S. Government.
The U.S. Government
is authorized to reproduce and distribute reprints for Government purposes notwithstanding any copyright
notation herein.}

\end{document}